\documentclass[journal,peerreview]{IEEEtran}

\usepackage[T1]{fontenc}
\usepackage{amsmath}
\usepackage{amssymb}
\usepackage{amsfonts}
\usepackage{amsthm}
\usepackage{float}
\usepackage[pdftex]{graphicx}
\usepackage{dsfont}
\usepackage{authblk}
\usepackage{comment}

\usepackage[style=ieee,sorting=nyt]{biblatex}
\theoremstyle{plain}
\newtheorem{thm}{Theorem}
\newtheorem{lem}[thm]{Lemma}
\newtheorem{col}[thm]{Corollary}
\newtheorem{pro}[thm]{Proposition}
\theoremstyle{remark}
\newtheorem*{rem}{Remark}

\newcommand{\1}{\mathds{1}}
\newcommand{\cX}{\mathcal{X}}
\newcommand{\cY}{\mathcal{Y}}
\newcommand{\cU}{\mathcal{U}}
\newcommand{\cT}{\mathcal{T}}
\newcommand{\cM}{\mathcal{M}}

\newcommand{\cS}{\mathcal{S}}

\newcommand{\cK}{\mathcal{K}}

\newcommand{\bX}{\bar{X}}
\newcommand{\bY}{\bar{Y}}
\newcommand{\bU}{\bar{U}}
\newcommand{\tX}{\tilde{X}}
\newcommand{\tY}{\tilde{Y}}
\newcommand{\tU}{\tilde{U}}

\newcommand{\tMo}{\tilde{M_{1}}}
\newcommand{\tMt}{\tilde{M_{2}}}

\newcommand{\supp}{\mathtt{supp}}

\def\BibTeX{{\rm B\kern-.05em{\sc i\kern-.025em b}\kern-.08em
		T\kern-.1667em\lower.7ex\hbox{E}\kern-.125emX}}

\usepackage{balance}

\begin{document}
\title{Tight Exponential Strong Converse for Source Coding Problem with Encoded Side Information
\thanks{Part of this paper was presented in 2023 IEEE International Symposium of Information Theory.}
}

%%%%%%
\author{
	Daisuke Takeuchi
    \IEEEmembership{Student Member, IEEE},
	Shun Watanabe
    \IEEEmembership{Senior Member, IEEE}
    \thanks{D. Takeuchi and S. Watanabe are with the Department of Electrical Engineering and Computer Science, Tokyo University of Agriculture and Technology, Japan. (e-mail: takehome@vivaldi.net; shunwata@cc.tuat.ac.jp)}
}

\maketitle

\begin{abstract}
The source coding problem with encoded side information is considered.
A lower bound on the strong converse exponent has been derived by Oohama, but its tightness has not been clarified. In this paper, we derive a tight strong converse exponent.
For the special case such that the side-information does not exists, we demonstrate that our tight exponent of the WAK problem reduces to the known tight expression of that special case while Oohama's lower bound is strictly loose.
The converse part is proved by a judicious use of the change-of-measure argument, which was introduced by Gu-Effros and further developed by Tyagi-Watanabe.
Interestingly, the soft Markov constraint, which was introduced by Oohama as a proof technique, is naturally incorporated into the characterization of the exponent.
A technical innovation of this paper is recognizing that the soft Markov constraint is a part of the exponent, rather than a penalty term that should be vanished.
In fact, via numerical experiment, we provide evidence that the soft Markov constraint is strictly positive.
The achievability part is derived by a careful analysis of the type argument; however, unlike the conventional analysis for the achievable rate region, we need to derive the soft Markov constraint in the analysis of the correct probability.
Furthermore, we present an application of our derivation of strong converse exponent to the privacy amplification.
\end{abstract}

\section{Introduction} \label{sec:intro}
\begin{figure}[h]
    \centering
    \includegraphics[scale=0.8]{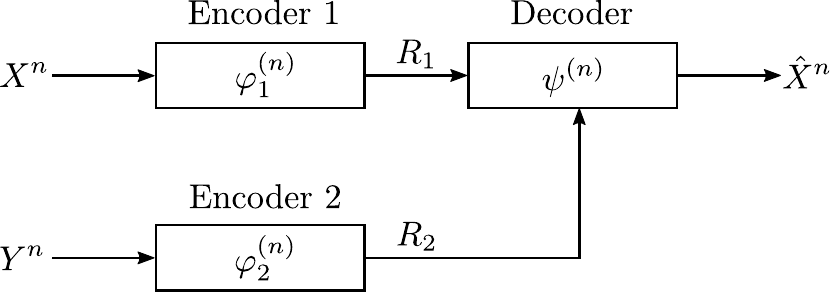}
    \caption{WAK network\protect\label{fig:WAK}}
\end{figure}
For source coding problems, the (weak) converse theorems state that, at transmission rates under the theoretical limit, there is no error-free coding scheme as the block length tends to infinity, i.e. the error probability cannot go to $0$.
On the other hand, the strong converse theorems state that, at transmission rates below the theoretical limit, the error probability converges to $1$ as the block length tends to infinity.

The source coding problem with encoded side information, also known as the Wyner-Ahlswede-Körner (WAK) problem, is one of fundamental problems in network information theory, and is described in Fig. \ref{fig:WAK}.
The WAK network consists of two encoders $\varphi^{(n)}_1,\varphi^{(n)}_2$ and one decoder $\psi^{(n)}$.
Both encoders are connected to the decoder by each channel, and the decoder aims to reproduce the source observed by the first encoder.
Under the condition of i.i.d. sources, this problem was solved by Wyner \cite{Wyn75}, and Ahlswede and Körner \cite{AK75} in the middle of 1970s.
According to their results, the information theoretical limit involves an auxiliary random variable with a Markov chain, which makes this problem difficult.

In this paper, we are interested in the strong converse theorem of WAK network.
As described later, it is known that for the WAK network, the error probability exponentially goes to $1$ when a rate pair is outside the achievable rate region.
On the other hand, an explicit form of the exponent function has not been clear.
Oohama \cite{Oo19} obtained a lower bound of the exponent by using the recursive method.
However, it was unknown whether Oohama's bound is tight or not.
Thus, we consider the exponent for the WAK network by using another method called the change of measure argument.
Furthermore, we discuss a property of the exponent through a numerical experiment, and we present an application of our derivation of strong converse exponent to one of the important concepts of cryptography: privacy amplification.

\subsection{Related Work}
Before stating our contribution, we would like to describe the history of exponential strong converse and change of measure argument.
For the WAK network, Ahlswede, Gács and Körner \cite{AGK76} proved that the correct probability goes to 0 if the rate pair is outside the rate region (\textit{strong converse}).
After that, Gu and Effros \cite{GE09} showed that the correct probability decreases exponentially by examining the proof of \cite{AGK76} (\textit{exponential strong converse}).
Since they just showed the convergence speed, they did not clarify the explicit form of the correct probability exponent in terms of the rate pair $(R_1,R_2)$.
Recently, Oohama \cite{Oo19} derived an explicit lower bound of this exponent function with parametric form.
However, it is unknown whether that bound is tight or not.

In 1973, Arimoto \cite{Ari73} analyzed the exponential strong converse for the point-to-point channel coding problem.
After that, in 1976, Dueck and Körner \cite{DK79} gave the strong converse exponent by using the type method.
Also, exponential strong converse has naturally been extended for multi-terminal networks.
For the Slepian-Wolf problem, Oohama and Han \cite{OH94} derived the optimal exponent by using universal coding in 1994.
Furthermore, as partially mentioned above, Oohama \cite{Oo16}, \cite{Oo19} derived the exponential strong converses for the WAK problem, the asymmetric broadcast channel problem \cite{KM77} and the Wyner-Ziv problem \cite{WZ75}.
However, due to Markov chains and an auxiliary random variable, the analysis of exponential strong converse for the WAK network is difficult and has not been studied much other than \cite{Oo19}.

The change of measure argument introduced by Gu and Effros \cite{GE09} is a useful technique for discussing the strong converses for various problems.
This enables us to prove strong converses using similar ways as those used for weak converses.
Then, it should be noted that Tyagi and Watanabe \cite{TW20} enabled the change of measure argument to be applied to a lot of distributed coding problems including the WAK network (see also an early attempt in \cite{Wat17b}).
After that, various coding problem (for example, \cite{SG20}, \cite{AWS23}) has been analyzed by using the change of measure argument.
More recently, Hamad, Wigger and Sarkiss \cite{HWM24} further developed the change of measure argument so that the Markov constraint can be proved in a different manner from \cite{TW20}.
In this paper, we derive a tight exponent of the correct probability for the WAK network by using the change of measure argument.
Particularly, our characterization clarifies a tight strong converse exponent for the problem such that there are both an auxiliary random variable and the Markov chain in the rate region.

For strong converses, various approaches have been proposed.
Kosut and Kliewer \cite{KK19} proposed an association between edge removal and strong converses, and also indicated that the capacity is not changed if an extra edge is removed.
Another related approach was proposed by Jose and Kulkarni \cite{JK17}, \cite{JK19}.
Their approach based on the linear programming (LP) analyzes and improves the converse bound for coding problems by introducing the so-called LP relaxation.
Furthermore, Liu, Handel and Verdú \cite{LHV19} developed an method for non-asymptotic converses by using reverse hypercontractivity of Markov semigroups.

\subsection{Main Contribution}
Looking at the above related works and remaining issues, we are motivated to consider the strong converse exponent for WAK network by judiciously using the change of measure argument.
In this work, we establish a tight characterization of the exponent.
Our main contributions are detailed as follows.

First, we apply the change of measure argument to the exponential strong converse problem for the source coding problem including the Markov chain and an auxiliary random variable.
Originally, the change of measure argument has been used in order to prove strong converses.
We now enables this technique to be applied to the analysis of strong converse exponents by focusing on the correct probability rather than achievable rates.
As mentioned above, a main difficulty to derive tight exponent is handling of the Markov chain constraint that appears in the characterization of the achievable rate region of the WAK problem.
Based on the idea of \cite{Oo18}, \cite{TW20} put the Markov penalty term in terms of the conditional mutual information in the correct probability exponent.
Rather than considering the Markov constraint as a penalty term that should be vanished, we considered regard it as a part of the exponent, which succeeded in deriving a tight exponent.
In fact, our numerical experiment suggests that the soft Markov constraint term is strictly positive.

Second, even though achievability part is proved using manipulations of the standard type method, proving the matching exponent is not trivial.
Particularly, as mentioned in the previous paragraph, our (converse) characterization of exponent may violate the Markov constraint.
However, in the WAK problem, the second encoder's coding can only depend on $Y^n$. Thus, we need to derive the Markov constraint term in the analysis of the correct probability.

Finally, our exponent is strictly tight compared to the previous bound \cite{Oo19} in a certain situation.
In the channel coding, there are two kinds of expressions for the strong converse exponent: the one by Arimoto \cite{Ari73} and the other by Dueck and Körner \cite{DK79}, which corresponds to the parametric function and the optimization one, respectively.
While the lower bound derived by Oohama \cite{Oo19} is closer to the expression by Arimoto, our characterization of the exponent is along the lines of the expression by Dueck and Körner.
Therefore, we cannot directly compare our exponent with Oohama's bound.
By converting our exponent expressed in an optimization problem into a parametric form, it become possible to analytically compare with the bound in \cite{Oo19} for the special case of single user source coding.

\subsection{Paper Organization}
This paper is organized as follows.
In Section \ref{sec:pre}, after introducing notations, we define the problem that we treat.
In Section \ref{sec:main}, we state our main result and provides some discussion of them.
In Section \ref{sec:comp}, we consider the comparison with Oohama's bound for the special case.
In Section \ref{sec:proof}, we give the proof of our result.
In Section \ref{sec:pa}, we consider an application to privacy amplification.
In Section VII, we conclude the paper.
\section{Preliminaries} \label{sec:pre}
\subsection{Notations}
Random variables are in capital case (e.g. $X$), and their realization are in lower case (e.g. $x$). All random variables take values in some alphabets that are in calligraphic letters (e.g. $\cX$).
We shall restrict our attention to finite alphabets only.
The cardinality of $\cX$ is denoted by $|\cX|$.
The probability distribution of random variable $X$ is denoted by $P_{X}$.
Random vectors and their realizations in the $n$th Cartesian product $\cX^n$ are denoted by $X^n=(X_{1},\dots,X_{n})$ and $x^n=(x_{1},\dots,x_{n})$, respectively.
Also, define $\mathcal{P}(\cX)$ as the set of all distribution on $\cX$.
For information quantities, we use the same notations as \cite{CK11}, such as the entropy $H(X)$, the mutual information $I(X\land Y)$ and the KL-divergence $D(P_X||Q_X)$.
For simplicity, we write $\cT_{X}$ for the type class of type $P_{X}$.
Also, we define the total variation distance $d_{\textnormal{var}}$ between two discrete distributions $P$ and $Q$ on a finite set $\cX$ as
\begin{align*}
d_{\textnormal{var}}(P,Q)=\frac{1}{2}\sum_{x\in\cX}|P(x)-Q(x)|.
\end{align*}
In the following, the base of logarithm and exponentiation is $2$.

\subsection{WAK Network}

We consider the network described in Fig. \ref{fig:WAK}. For a given
source $P_{XY}$ on a finite alphabet $\cX\times\cY $,
this coding system consists of two encoders
\begin{align*}
\varphi_{1}^{(n)}:\cX^n\to\cM_{1}^{(n)},\ \varphi_{2}^{(n)}:\cY^n\to\cM_{2}^{(n)},
\end{align*}
and one decoder
\begin{align*}
\psi^{(n)}:\cM_{1}^{(n)}\times\cM_{2}^{(n)}\to\cX^n.
\end{align*}
For simplicity, we write encoders and a decoder as $\varphi_{1},\varphi_{2}$
and $\psi$, respectively. Input sequences $X^n$ and $Y^n$ are
encoded by $\varphi_{1}$ and $\varphi_{2}$ separately, and the decoder
outputs the estimation $\hat{X}^n=\psi(\varphi_{1}(X^n),\varphi_{2}(Y^n))$
of $X^n$.

For i.i.d. $(X^n,Y^n)\sim P_{X^nY^n}$, the error probability of code $\Phi_{n}=(\varphi_{1},\varphi_{2},\psi)$ is defined as
\begin{align*}
P_{\textnormal{e}}(\Phi_{n}|P_{X^nY^n}):=\Pr(\hat{X}^n\neq X^n).
\end{align*}
A rate pair $(R_{1},R_{2})$ is called an achievable rate pair if there exists a sequence of codes such that
\begin{align*}
\limsup_{n\to\infty}\frac{1}{n}\log|\cM_{1}^{(n)}|\leq R_{1},\ \limsup_{n\to\infty}\frac{1}{n}\log|\cM_{2}^{(n)}|\leq R_{2},
\end{align*}
and
\begin{align*}
\lim_{n\to\infty}P_{\textnormal{e}}(\Phi_{n}|P_{X^nY^n})=0.
\end{align*}
Also, the achievable rate region $\mathcal{R}_{\textnormal{WAK}}$ is denoted by
\begin{align*}
\mathcal{R}_{\textnormal{WAK}}=\{(R_{1},R_{2}):(R_{1},R_{2})\text{ is achievable}\}.
\end{align*}
Wyner \cite{Wyn75}, Ahlswede and Körner \cite{AK75} determined $\mathcal{R}_{\textnormal{WAK}}$ as below.
\begin{thm}
	Consider the region
	\begin{align}
		\mathcal{R}=\bigcup_{\substack{
				X-Y-U\\
				|\cU|\leq|\cY|+1
		}}\{(R_{1},R_{2}):&R_{1}\geq H(X|U),R_{2}\geq I(Y\land U)\},
	\end{align}
	where the Markov chain $U-Y-X$ means
	\begin{align*}
	P_{UXY}=P_{X}P_{Y|X}P_{U|Y}.
	\end{align*}
	Then, it holds
	\begin{align}
	\mathcal{R}_{\textnormal{WAK}}=\mathcal{R}.
	\end{align}
\end{thm}

Our main purpose is to find an explicit form of the exponent function for the error probability converging to $1$ as $n\to\infty$ when $(R_{1},R_{2})$ is not achievable. To achieve this aim, we define the correct probability $P_{\textnormal{c}}$ and the optimal exponent function $G$:
\begin{align*}
	P_{\textnormal{c}}(\Phi_{n}|P_{X^nY^n}) & =1-P_{\textnormal{e}}(\Phi_{n}|P_{X^nY^n}),\\
	G^{(n)}(R_{1},R_{2}|P_{X^nY^n}) & =\min_{\Phi_{n}}\left(-\frac{1}{n}\right)\log P_{c}(\Phi_{n}|P_{X^nY^n})\\
	G(R_{1},R_{2}|P_{XY}) & =\lim_{n\to\infty}G^{(n)}(R_{1},R_{2}|P_{X^nY^n}).
\end{align*}
According to \cite{Oo19}, the limit in the definition of $G$ exists, and the exponent function $G(R_{1},R_{2}|P_{XY})$ is a convex function of $(R_{1},R_{2})$.

Ahlswede, Gács and Körner \cite{AGK76} proved that the correct probability $P_{\textnormal{c}}$ goes to $0$ as $n\to\infty$.
Gu and Effros \cite{GE09} showed that the convergence speed of $P_{\mathrm{c}}$ is exponential. Also, Oohama \cite{Oo19} derived a lower bound on $G$.
However, it was not clarified whether the lower bound in \cite{Oo19} is tight or not.
In the next section, we shall show our result.

\section{Main Result} \label{sec:main}
In this section, we show our main result. Before showing it, we set
\begin{align}
	&F(R_{1},R_{2}|P_{XY}) = \min_{P_{\tU \tX \tY }\in\mathcal{P}(\cU \times\cX \times\cY )}\biggl\{D(P_{\tU \tX \tY }||P_{\tU |\tY }P_{XY})+\left|I(\tU \land\tY )-R_{2}\right|^{+}:R_{1}\geq H(\tX|\tU),|\cU |\leq|\cX||\cY |+2\biggr\},\label{eq:formula_R}
\end{align}
where $|a|^{+}=\max\{a,0\}$.

Now we show our main theorem by using this quantity.
\begin{thm}
	For any $P_{XY}$, we have
	\begin{align}
	G(R_{1},R_{2}|P_{XY})=F(R_{1},R_{2}|P_{XY}).
	\end{align}
    \label{thm:main}
\end{thm}

The proof of this result can be found in Section \ref{sec:proof}.
As stated in Section \ref{sec:intro}, there are two ways of expressing the exponent: parametric \cite{Ari73} and optimization \cite{DK79} function.
Unlike Oohama's bound with parametric notations \cite{Oo19}, we derive the exponent as an optimization function.
Interestingly, although the Markov chain in the achievable rate region cannot be explicitly seen in \eqref{eq:formula_R}, the conditional mutual information in $D(P_{\tU \tX \tY }||P_{\tU |\tY }P_{XY})$ plays the role.
In fact, we can break it into the KL-divergence and the conditional mutual information as
\begin{align}
    D(P_{\tU \tX \tY }||P_{\tU |\tY }P_{XY}) &= D(P_{\tX \tY }||P_{XY})+D(P_{\tU |\tX \tY }||P_{\tU |\tY }|P_{\tX \tY }) \nonumber\\
    &= D(P_{\tX \tY }||P_{XY})+I(\tU \land\tX |\tY ).
    \label{eq:divandI1}
\end{align}
Then, we get $I(\tU \land \tX|\tY) = 0$ if the Markov chain is satisfied, which corresponds to replacing the ``hard'' information constraint with the ``soft'' cost introduced by Oohama \cite{Oo18}, \cite{Oo19}.

Moreover, the exponent $F(R_{1},R_{2}|P_{XY})$ is strictly positive when the rate pair $(R_1,R_2)$ is outside the achievable rate region $\mathcal{R}_{\textnormal{WAK}}$.
Actually, we can verify this property by considering the contraposition that $(R_1,R_2)$ is inside the achievable rate region when $F(R_{1},R_{2}|P_{XY}) = 0$ that is derived from the Markov constraint $I(\tU \land \tX|\tY) = 0$ and the fact $P_{\tX\tY} = P_{XY}$ since $D(P_{\tX\tY}||P_{XY}) = 0$.

In the following, we investigate our result to observe its properties.

\subsection{Numerical Experiment}
In this section, we make numerical calculation in a specific situation.
Since our exponent has an auxiliary random variable, it is not easy to obtain the optimal value.
Even though it is difficult to exactly evaluate the characterization of our exponent, in the following, we evaluate the exponent for a natural choice of auxiliary distribution to get some insight.
We investigate the properties of the exponent by considering a simple case, doubly symmetric binary source (DSBS).

Consider a DSBS with probability $p \in (0,1/2)$, whose distribution is denoted by DSBS$(p)$.
Then, let $P_{XY} = \textnormal{DSBS}(p)$.
In other words, $X$ is a Bernoulli random variable with parameter $1/2$, and $Y$ is the output of a binary symmetric channel BSC$(p)$ with crossover probability $p$ when the input is $X$.
Also, fix the auxiliary alphabet as $|\cU| = 2$.
Furthermore, for the purpose of numerical calculation, we choose $P_{\tU\tY} = \textnormal{DSBS}(\beta)$ and
\begin{align*}
	P_{\tX|\tU\tY}(x|u,y) = \begin{cases}
		q_0 & (u = y \textnormal{ and } x \neq y), \\
		1-q_0 & (u = y \textnormal{ and } x = y), \\
		q_1 & (u \neq  y \textnormal{ and } x \neq y), \\
		1-q_1 & (u \neq  y \textnormal{ and } x = y),
	\end{cases}
\end{align*}
for $q_0, q_1 \in (0,1/2)$.
In this situation, we get
\begin{align*}
	P_{\tX|\tU}(x|u) = \begin{cases}
		(1-\beta)(1-q_0)+\beta q_1 & (u = x), \\
		\beta (1-q_1)+(1-\beta)q_0 & (u \neq x) .
	\end{cases}
\end{align*}
Therefore, we can see
\begin{align*}
	H(\tX|\tU) &= h((1-\beta)(1-q_0)+\beta q_1),\\
	I(\tU \land \tY) &= 1-h(\beta),\\
	D(P_{\tU\tX\tY}||P_{\tU|\tY}P_{XY}) &= D(P_{\tY}||P_Y)+D(P_{\tX|\tU\tY}||P_{X|Y}|P_{\tU\tY})\\
	&=(1-\beta)D(q_0||p)+\beta D(q_1||p),
\end{align*}
where $h(\alpha)$ is a binary entropy function and $D(x||y) = x\log (x/y)+(1-x)\log ((1-x)/(1-y))$.
Thus, we can derive an upper bound of the strong converse exponent as
\begin{align}
	\min_{(\beta,q_0,q_1) \in [0,1]^3}\{(1-\beta)D(q_0||p)+\beta D(q_1||p)+|1-h(\beta)-R_2|^+:h((1-\beta)(1-q_0)+\beta q_1) \leq R_1\}.
	\label{eq:exp_bound}
\end{align}
Note that the Markov chain $\tU-\tY-\tX$ holds if and only if $q_0 = q_1$.

\begin{figure}[h]
	\centering
	\includegraphics[scale=0.8]{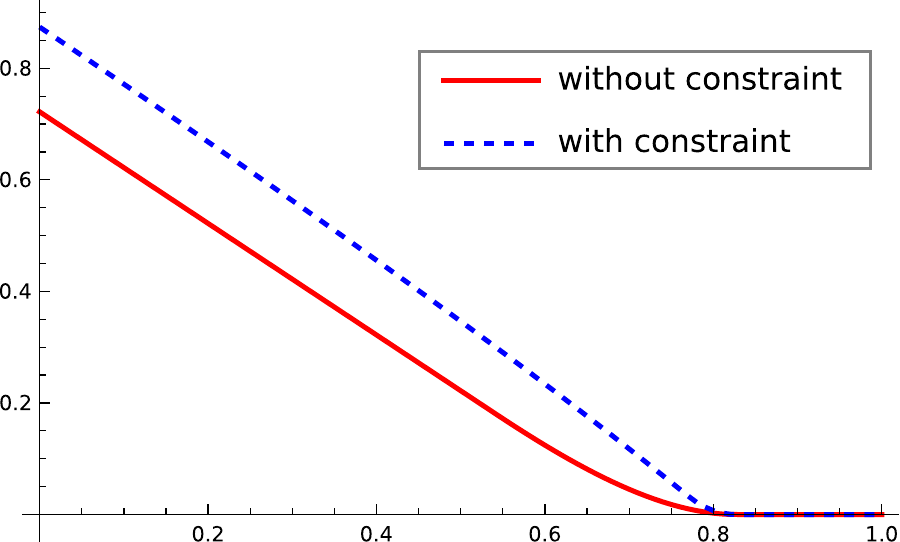}
	\caption{Graphs of \eqref{eq:exp_bound} for $(p, R_2) = (0.1,1-h(0.2))$.
	The vertical axis shows the value of \eqref{eq:exp_bound} and the horizontal axis shows the value of $R_1$.
	The solid line indicates the value without the Markov constraint for $P_{\tU\tX\tY}$, and the dashed line indicate the value with the constraint.
	\protect\label{fig:graph}
	}
\end{figure}
Now we compute \eqref{eq:exp_bound} with numerical experiment.
Figure \ref{fig:graph} shows the plots of the values of \eqref{eq:exp_bound} for $(p, R_2) = (0.1,1-h(0.2))$ in the range $R_1 \in [0,1]$, where the vertical axis is the value of \eqref{eq:exp_bound} and the horizontal axis is the value of $R_1$.
Also, the solid line indicates the value without the Markov constraint for $P_{\tU\tX\tY}$, and the dashed line indicate the value with the constraint, i.e., $q_0 = q_1$.

In Figure \ref{fig:graph}, we can see that the value are smaller when there is not the Markov constraint.
We recall that $I(\tU \land \tX|\tY) = 0$ corresponds to a Markov constraint for our exponent.
Furhermore, we recall that the auxiliary random variable must satisfy the Markov constriant in the characterization of the achievable rate region.
Therefore, one might expect that the optimal value of \eqref{eq:exp_bound} is attained when the Markov constraint is satisfied, i.e., when $q_0 = q_1$.
However, Figure \ref{fig:graph} indicates the different result, which implies that the Markov constraint for $P_{\tU\tX\tY}$ doesn't necessarily yield the optimal exponent.

\subsection{Network with Non-encoded side information}
\begin{figure}[h]
	\centering
	\includegraphics[scale=0.8]{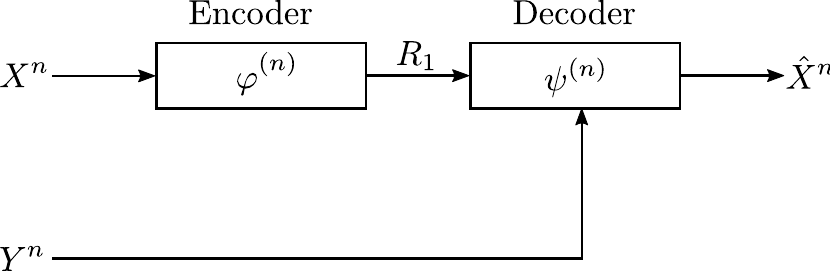}
	\caption{Network with non-encoded side information\protect\label{fig:FS}}
\end{figure}
Let us consider a sub-problem of WAK coding (also Slepian-Wolf coding), a network with non-coded side information described in Fig. \ref{fig:FS}.
This network corresponds to the WAK network such that $R_2 \geq \log |\cY|$.

For this network, we define the optimal exponent function $G_\textnormal{NE}$:
\begin{align*}
	P_{\textnormal{c}}(\Phi_{n}|P_{X^nY^n}) & =1-P_{\textnormal{e}}(\Phi_{n}|P_{X^nY^n}),\\
	G_\textnormal{NE}^{(n)}(R_1|P_{X^nY^n}) & =\min_{\Phi_{n}}\left(-\frac{1}{n}\right)\log P_{c}(\Phi_{n}|P_{X^nY^n}),\\
	G_\textnormal{NE}(R_1|P_{XY}) & =\lim_{n\to\infty}G_\textnormal{NE}^{(n)}(R_1|P_{X^nY^n}),
\end{align*}
and we set
\begin{align}
	&F_\textnormal{NE}(R_1|P_{XY}) = \min_{P_{\tX\tY}\in\mathcal{P}(\cX\times\cY)} \biggl\{D(P_{\tX\tY}||P_{XY})+\left|H(\tX|\tY)-R_1\right|^{+}\biggr\}.
	\label{eq:formula_FS}
\end{align}
Then, by using this quantity, we have the exponent.
For any $P_{XY}$, we can obtain the following corollary.
\begin{col}\label{thm:fs}
	For $R_2 \geq \log |\cY|$, we have
	\begin{align}
		G(R_{1},R_{2}|P_{XY}) = F_\textnormal{NE}(R_1|P_{XY}).
	\end{align}
\end{col}
\begin{IEEEproof}
	To prove Corollary \ref{thm:fs}, we shall show two inequalities:
	\begin{align}
		G(R_{1},R_{2}|P_{XY}) \geq F_\textnormal{NE}(R_1|P_{XY}), \label{eq:fs_geq} \\
		G(R_{1},R_{2}|P_{XY}) \leq F_\textnormal{NE}(R_1|P_{XY}). \label{eq:fs_leq}
	\end{align}
	For \eqref{eq:fs_geq}, let $\tU =\tY$ for $F(R_{1},R_{2}|P_{XY})$ , and this derives
	\begin{align}
		G(R_{1},R_{2}|P_{XY}) &= F(R_{1},R_{2}|P_{XY}) \nonumber \\
		&\geq \min_{\substack{
				P_{\tU \tX \tY }\in\mathcal{P}(\cU \times\cX \times\cY )\\
				\tU = \tY
		}}\biggl\{D(P_{\tU \tX \tY }||P_{\tU |\tY }P_{XY})+\left|I(\tU \land\tY )-R_{2}\right|^{+}:R_{1}\geq H(\tX|\tU),|\cU |\leq|\cX||\cY |+2\biggr\} \nonumber \\
		&= \min_{P_{\tX\tY}\in\mathcal{P}(\cX\times\cY)} \biggl\{D(P_{\tX\tY}||P_{XY}):R_{1}\geq H(\tX|\tY)\biggr\}. \nonumber \\
		&= \min_{P_{\tX\tY}\in\mathcal{P}(\cX\times\cY)} \biggl\{D(P_{\tX\tY}||P_{XY})+\left|H(\tX|\tY)-R_1\right|^{+}\biggr\},
	\end{align}
	where the last equality can be obtained by considering similarly to part of \cite[Problem 2.6]{CK11}.

	On the other hand, for \eqref{eq:fs_leq}, when $R_1 \leq H(\tX|\tY)$, the KL-divergence $D(P_{\tU\tX\tY}||P_{\tU|\tY}P_{XY})$ can be evaluated as
	\begin{align}
		D(P_{\tU\tX\tY}||P_{\tU|\tY}P_{XY}) &= D(P_{\tX\tY}||P_{XY})+I(\tU \land \tX|\tY) \nonumber\\
		&= D(P_{\tX\tY}||P_{XY})+H(\tX|\tY)-H(\tX|\tU,\tY) \nonumber\\
		&\geq D(P_{\tX\tY}||P_{XY})+H(\tX|\tY)-R_1 \nonumber\\
		&= D(P_{\tX\tY}||P_{XY})+\left|H(\tX|\tY)-R_1\right|^{+},
		\label{eq:conFS1}
	\end{align}
	where the inequality is derived by the condition of $R_1$:
	\begin{align*}
		R_1 &\geq H(\tX|\tU) \geq H(\tX|\tU,\tY).
	\end{align*}
	Also, when $R_1 \geq H(\tX|\tY)$, we can see
	\begin{align}
		D(P_{\tU\tX\tY}||P_{\tU|\tY}P_{XY}) &\geq D(P_{\tX\tY}||P_{XY}), \nonumber \\
		&= D(P_{\tX\tY}||P_{XY})+\left|H(\tX|\tY)-R_1\right|^{+},
		\label{eq:conFS2}
	\end{align}
	and we combine \eqref{eq:conFS1} and \eqref{eq:conFS2}, then  we obtain
	\begin{align}
		&G(R_{1},R_{2}|P_{XY}) \geq \min_{P_{\tX\tY}\in\mathcal{P}(\cX\times\cY )}\biggl\{D(P_{\tX\tY}||P_{XY})+\left|H(\tX|\tY)-R_1\right|^{+}\biggr\},
	\end{align}
	which completes the proof.
\end{IEEEproof}

\section{Comparison to Oohama's bound} \label{sec:comp}
In this chapter, we now compare Oohama's bound and our exponent for the simplest network, i.e., the single-user network.
Before showing it, we introduce Oohama's expression as below \cite{Oo19}.

\begin{pro}\label{thm:oohama_bound}
    Set
    \begin{align*}
    \tilde{\mathcal{P}}(P_{XY}) \!=\! \{
    P_{\tU\tX\tY} : |\cU| \leq |\cY|,U-Y-X,P_{X|Y} = P_{\tX|\tY}
    \}.
    \end{align*}
    For $(\mu,\alpha) \in [0,1]^2$ and for $P_{\tU\tX\tY} \in \tilde{\mathcal{P}}(P_{XY})$, define
    \begin{align*}
        \tau_{\tU\tX\tY|XY}^{(\mu,\alpha)}(x,y|u) &= (1-\alpha)\log\frac{P_{\tY}(y)}{P_Y(y)}+\alpha\left(\mu\log\frac{P_{\tY|\tU}(y|u)}{P_Y(y)}+(1-\mu)\log\frac{1}{P_{\tX|\tU}(x|u)}\right), \\
        \omega^{(\mu,\alpha)}(P_{\tU\tX\tY}|P_{XY}) &= -\log \sum_{(u,x,y)} P_{\tU\tX\tY}(u,x,y)\cdot\exp(-\tau_{\tU\tX\tY|XY}^{(\mu,\alpha)}(x,y|u)), \\
        \Omega^{(\mu,\alpha)}(P_{XY}) &= \min_{P_{\tU\tX\tY}}\omega^{(\mu,\alpha)}(P_{\tU\tX\tY}|P_{XY}), \\
        f_\textnormal{o}^{(\mu,\alpha)}(R_1,R_2|P_{XY}) &\!=\! \frac{\Omega^{(\mu,\alpha)}(P_{XY})\!-\!\alpha((1-\mu)R_1+\mu R_2)}{2+\alpha(1-\mu)} \\
        F_\textnormal{o}(R_1,R_2|P_{XY}) &= \sup_{(\mu,\alpha)\in[0,1]^2} f_\textnormal{o}^{(\mu,\alpha)}(R_1,R_2|P_{XY}).
    \end{align*}
    Then, for any $P_{XY}$, we have
    \begin{align}
    P_{\textnormal{c}}(\Phi_{n}|P_{X^nY^n}) \leq 5\exp(-nF_\textnormal{o}(R_1,R_2|P_{XY})).
    \end{align}
\end{pro}

\begin{rem}
    Oohama showed that $F_\textnormal{o}(R_1,R_2|P_{XY})$ is strictly positive if $(R_1,R_2)$ is outside the rate region $\mathcal{R}_{\textnormal{WAK}}$.
\end{rem}

\subsection{Application to expressions to the single-user network}
\begin{figure}[h]
    \centering
    \includegraphics[scale=0.8]{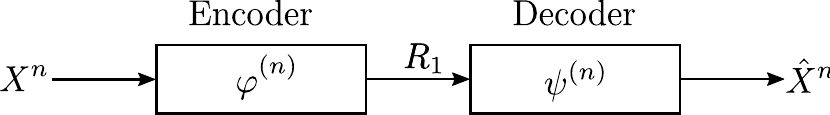}
    \caption{Single-user network\protect\label{fig:single}}
\end{figure}

In this study, we consider two expressions for the single-user network (Fig. \ref{fig:single}).
This network corresponds to the WAK network that the encoder $\varphi_2$ sends nothing or $|\cY| = 1$.
In this case,  we can simplify Oohama's bound (Corollary \ref{col:single_bound}) and our exponent (Corollary \ref{col:single_exponent}) as follows.

\begin{col}\label{col:single_bound}
    For $(\mu,\alpha) \in [0,1]^2$, define
    \begin{align*}
        \hat\tau_{X}^{(\mu,\alpha)}(x) &= \alpha(1-\mu)\log\frac{1}{P_X(x)}, \\
        \hat\Omega^{(\mu,\alpha)}(P_X) &= -\log \sum_x P_X(x) \exp(-\hat\tau_X^{(\mu,\alpha)}(x)), \\
        \hat f_\textnormal{o}^{(\mu,\alpha)}(R_1|P_X) &= \frac{\hat\Omega^{(\mu,\alpha)}(P_X)-\alpha(1-\mu)R_1}{2+\alpha(1-\mu)} \\
        \hat F_\textnormal{o}(R_1|P_X) &= \max_{(\mu,\alpha)\in[0,1]^2} \hat f_\textnormal{o}^{(\mu,\alpha)}(R_1|P_X).
    \end{align*}
    Then, for any $P_{XY}$, we have
    \begin{align}
    F_\textnormal{o}(R_1,R_2|P_{XY}) = \hat F_\textnormal{o}(R_1|P_X).
    \end{align}
\end{col}

Taking $|\cY| = 1$ in Corollary \ref{thm:fs}, we have the following corollary.
\begin{col}\label{col:single_exponent}
    Set
    \begin{align}
        \hat F_\textnormal{t}(R_1|P_X) = \min_{P_{\tX}\in\mathcal{P}(\cX)} \biggl\{
        D(P_{\tX}||P_X)+\left|H(\tX)-R_1\right|^+
        \biggr\}.
    \end{align}
    Then, for any $P_{XY}$, we have
    \begin{align}
    F(R_{1},R_{2}|P_{XY}) = \hat F_\textnormal{t}(R_1|P_X).
    \end{align}
\end{col}

\subsection{Comparison}
In this section, we compare two expressions in Corollary \ref{col:single_bound} and Corollary \ref{col:single_exponent}.
If we can verify Proposition \ref{pro:main_result} in the following, our exponent is strictly tight for the single-user network.

\begin{pro}\label{pro:main_result}
    For the special case described in Fig. \ref{fig:single} and $R_1 < H(X)$, we have
    \begin{align}
    \hat F_\textnormal{o}(R_1|P_X) < \hat F_\textnormal{t}(R_1|P_X).
    \end{align}
\end{pro}

As mentioned in Introduction, two expressions are different in characterization.
Since we cannot compare them as they are, we convert our exponent using \cite[Exercise 2.41]{Hay06}, which is stated in the following.

\begin{lem}\label{lem:convert}
    We can rewrite $\hat F_{\textnormal{t}}(R_1|P_X)$ as
    \begin{align}
        \hat F_\textnormal{t}(R_1|P_X) &= \max_{\theta \leq 0}\frac{-s(\theta)+\theta R_1}{1-\theta},
    \end{align}
    where the function $s(\theta)$ is denoted as
    \begin{align*}
    s(\theta) = \log \sum_x \{P_X(x)\}^{1-\theta}.
    \end{align*}
\end{lem}

Using this lemma, we prove the main result, Proposition \ref{pro:main_result}.

\begin{IEEEproof}[Proof of Proposition \ref{pro:main_result}]
    Putting $\theta = -\alpha(1-\mu) \in [-1,0]$, we first see that
    \begin{align}
    \hat\Omega^{(\mu,\alpha)}(P_X) = -s(\theta),
    \end{align}
    which can be derived from
    \begin{align}
        s(\theta) &= \log \sum_x \{P_X(x)\}^{1-\theta} \nonumber\\
        &= \log \left(\sum_x P_X(x)\left(\frac{1}{P_X(x)}\right)^\theta\right) \nonumber\\
        &= \log \left(\sum_x P_X(x) \exp\left(\theta \log\frac{1}{P_X(x)}\right)\right) \nonumber\\
        &= -\hat\Omega^{(\mu,\alpha)}(P_X).
    \end{align}
    Then, we can rewrite $\hat F_\textnormal{o}(R_1|P_X) $ as
    \begin{align}
    \hat F_\textnormal{o}(R_1|P_X) = \max_{\theta \in [-1,0]} \frac{-s(\theta)+\theta R_1}{2-\theta}.
    \end{align}
    Now we show that for some $\theta \in [-1,0)$, the numerator $-s(\theta)+\theta R_1$ is positive, i.e., $-s(\theta)+\theta R_1 > 0$.
    Considering the condition of $R_1$: $R_1 < H(X)$, we rewrite $R_1 \leq H(X)-\delta$ for a small positive number $\delta$.
    Also, from the definition of derivatives, we get
    \begin{align}
        \frac{d}{d\theta}s(0) = \lim_{\theta \to 0} \frac{s(\theta)}{\theta} = H(X).
    \end{align}
    Then, when we take $\theta_0 \in [-1,0)$ sufficiently close to 0, there exists $\delta$ such that
    \begin{align}
        \frac{s(\theta_0)}{\theta_0} \geq H(X)-\frac{\delta}{2},
    \end{align}
    which derives
    \begin{align}
        R_1 \leq H(X)-\delta \leq \frac{s(\theta_0)}{\theta_0}-\frac{\delta}{2} < \frac{s(\theta_0)}{\theta_0}.
    \end{align}
    Therefore, noting that $\theta_0$ is a negative number, we have showed that $-s(\theta)+\theta R_1 > 0$ for some $\theta \in [-1,0)$, which indicates that there exists $\theta_1 \in [-1,0)$ such that
    \begin{align}
    	F_\textnormal{o}(R_1|P_X) = \frac{-s(\theta_1)+\theta_1 R_1}{2-\theta_1}.
    \end{align}
    Thus, we have
    \begin{align}
        F_\textnormal{o}(R_1|P_X) &= \frac{-s(\theta_1)+\theta_1 R_1}{2-\theta_1} \nonumber \\
        &< \frac{-s(\theta_1)+\theta_1 R_1}{1-\theta_1} \nonumber \\
        &\leq \max_{\theta \leq 0}\frac{-s(\theta)+\theta R_1}{1-\theta} \nonumber\\
        &= F_\textnormal{t}(R_1|P_X),
    \end{align}
    which implies Proposition \ref{pro:main_result}.
\end{IEEEproof}

\section{Proof of the Main Theorem} \label{sec:proof}
In this section, we shall prove Theorem \ref{thm:main}.
\subsection{Direct Part}
In this section, we prove the direct part. Since we are interested
in a lower bound on the correct probability $P_{\mathrm{c}}$, it
suffices to construct a code such that sequences $(x^n,y^n)$
in a fixed type class ${\cal T}_{\bX\bY}$ can be decoded
correctly. First, we fix $P_{\tU \tX \tY }$ that attains
the minimum in (\ref{eq:formula_R}) for $(R_{1}-2\epsilon,R_{2})$
where $\epsilon$ is a small positive number. Let $P_{\bU \bX\bY}$
be a joint type satisfying
\begin{align}
d_{\textnormal{var}}(P_{\bU \bX\bY},P_{\tU \tX \tY })\leq\frac{|\cU ||\cX ||\cY |}{n}.
\end{align}
Note that, by the continuity of entropy, we see
\begin{align}
	H(\bX|\bU ) & \leq H(\tX |\tU )+\epsilon \nonumber\\
	& \leq R_{1}-\epsilon,
\end{align}
for sufficiently large $n$. In order to prove the direct part, we
use the following covering lemma.
\begin{lem}
	\label{lem:covering}For a joint type $P_{\bU \bY}$, there
	exists a codebook ${\cal C}=\{u_{1}^n,\ldots,u_{L}^n\}\subseteq{\cal T}_{\bU}$
	such that
	\begin{equation}
		L\le\exp\{nI(\bU \land\bY)+(|\cU ||\cY |+4)\log(n+1)\},\label{eq:convering}
	\end{equation}
	and such that for any $y^n\in\cT _{\bY}$, there exists
	$u^n\in\mathcal{C}$ satisfying $(u^n,y^n)\in\cT _{\bU \bY}$.
\end{lem}

\begin{IEEEproof}
	Since the strategy of proving this lemma is almost the same as \cite[Lemma 4]{Wat17a},
	we omit it.
\end{IEEEproof}
We describe an overview of code construction. For this coding scheme,
we want to evaluate the correct probability $P_{\textnormal{c}}$. For
$\varphi_{1}$, we use the standard random binning with rate $R_{1}$.
For $\varphi_{2}$, we use the quantization using the test channel
$P_{\bU |\bY}$. However, $R_{2}$ may be smaller than $I(\bU \land\bY)$
that is needed to quantize sequences in $\cT _{\bY}$ via
the test channel $P_{\bU |\bY}$. In that case, $\varphi_{2}$
assigns messages to a part of the quantization codebook, and assign
a prescribed constant message to the remaining. The detail of the
coding scheme is described as follows.

\paragraph{Encoding of $\varphi_{1}$}

Randomly and independently assign an index $m_{1}=F(x^n)\in\{1,2,\dots,2^{nR_{1}}\}$
to each sequence $x^n\in\cX ^n$. The set of sequences
with the same index $m_{1}$ forms a bin $\mathcal{B}(m_{1})$. Observing
$x^n\in\mathcal{B}(m_{1})$, $\varphi_{1}$ sends the bin index
$m_{1}$.

\paragraph{Encoding of $\varphi_{2}$}

Consider a codebook ${\cal C}$ given by Lemma \ref{lem:covering}
and function $f:{\cal T}_{\bY}\to\{1,\ldots,L\}$ such that $(u_{f(y^n)}^n,y^n)\in{\cal T}_{\bU \bY}$.
Without loss of generality, we assume that $|f^{-1}(1)|\ge|f^{-1}(2)|\ge\cdots\ge|f^{-1}(L)|$.
We define a function $\tilde{f}$ as
\begin{align*}
\tilde{f}(y^n)=\begin{cases}
	f(y^n) & (f(y^n)\leq2^{nR_{2}}),\\
	1 & (f(y^n)>2^{nR_{2}}).
\end{cases}
\end{align*}
Then, upon observing $y^n\in{\cal Y}^n$, $\varphi_{2}$ sends
$m_{2}=\tilde{f}(y^n)$ if $y^n\in{\cal T}_{\bY}$ and $m_{2}=1$
otherwise.

\paragraph{Decoding.}

Observing $(m_{1},m_{2})=(\varphi_{1}(x^n),\varphi_{2}(y^n))$,
$\psi$ declares the estimate $\hat{x}^n\in\cT _{\bX|\bU }(u_{m_{2}}^n)$
if it is the unique sequence satisfying $F(\hat{x}^n)=m_{1}$; otherwise
it declares a prescribed sequence.

\paragraph{Analysis of the correct probability.}

We bound the correct probability $P_{\mathrm{c}}$ averaged over random
bin assignments. We lower bound the correct probability quite generously,
and estimate the probability such that the following three events
occur consecutively.
\begin{enumerate}
	\setlength{\leftskip}{-0.1cm}
	\item The event $E_{1}$ such that $y^n\in\cT_{\bY}$ and $\tilde{f}(y^n)=f(y^n)$:
	we denote the probability of this event by $P_{\textnormal{c},1}=\Pr(E_{1})$.
	\item The event $E_{2}$ such that $E_{1}$ occurs and $x^n$ is included
	in the conditional type class ${\cal T}_{\bX|\bU \bY}(u^n_{f(y^n)},y^n)$:
	we denote the conditional probability of this event conditioned that
	$E_{1}$ occurs by $P_{\textnormal{c},2}=\Pr(E_{2}|E_{1})$.
	\item The event $E_{3}$ such that $E_{1}\wedge E_{2}$ occurs and there
	exists no sequence $\hat{x}^n\in{\cal T}_{\bX|\bU }(u_{f(y^n)}^n)$
	satisfying $\hat{x}^n\neq x^n$ and $F(\hat{x}^n)=F(x^n)$:
	we denote the conditional probability of this event conditioned on
	the first and second events by $P_{\textnormal{c},3}=\Pr(E_{3}|E_{1}\wedge E_{2})$.
\end{enumerate}
By noting that ${\cal T}_{\bX|\bU \bY}(u_{f(y^n)}^n,y^n)\subseteq{\cal T}_{\bX|\bU }(u_{f(y^n)}^n)$,
we can verify that the decoder $\psi$ outputs the correct source
sequence $x^n$ if the three events occur $E_{1},E_{2}$, and $E_{3}$
consecutively. Thus, we can lower bound the correct probability as
\begin{align}
P_{\textnormal{c}}\geq\prod_{i=1}^{3}P_{\textnormal{c},i}.
\end{align}
In the following, we shall calculate each probability $P_{\textnormal{c},i}$.

When $L\le2^{nR_{2}}$, the second condition $\tilde{f}(y^n)=f(y^n)$
of the event $E_{1}$ is always satisfied. On the other hand, by noting
that $|f^{-1}(i)|$ is sorted in descending order, when $L>2^{nR_{2}}$,
we have
\begin{align}
\frac{1}{2^{nR_{2}}}\sum_{i=1}^{2^{nR_{2}}}|f^{-1}(i)|\geq\frac{1}{L}\sum_{i=1}^{L}|f^{-1}(i)|=\frac{|{\cal T}_{\bY}|}{L},
\end{align}
which implies
\begin{align}
\sum_{i=1}^{2^{nR_{2}}}\frac{|f^{-1}(i)|}{|{\cal T}_{\bY}|}\geq\frac{2^{nR_{2}}}{L}.
\end{align}
Therefore, by using \cite[Lemma 2.6]{CK11}, we have
\begin{align}
	P_{\textnormal{c},1} & =\Pr\{Y^n\in{\cal T}_{\bY}\land[f(Y^n)\leq2^{nR_{2}}]\}\nonumber \\
	& =P_{Y^n}({\cal T}_{\bY})\sum_{i=1}^{2^{nR_{2}}}\frac{|f^{-1}(i)|}{|{\cal T}_{\bY}|}\nonumber \\
	& \geq P_{Y^n}({\cal T}_{\bY})\min\left\{1,\frac{2^{nR_{2}}}{L}\right\}\nonumber \\
	& \geq(n+1)^{-|\cY |}2^{-n(D(P_{\bY}||P_{Y})+|(1/n)\log L-R_{2}|^{+})}\nonumber \\
	& \geq(n+1)^{-|\cY |}2^{-n(D(P_{\bY}||P_{Y})+|I(\bU \land\bY)-R_{2}+\delta_{n}|^{+})},
\end{align}
where $\delta_{n}=(1/n)(|\cU ||\cY |+4)\log(n+1)$.

By using \cite[Lemma 2.6]{CK11}, we have
\begin{align}
	P_{\textnormal{c},2} & =\sum_{\substack{
			y^n\in{\cal T}_{\bY}\\
			\tilde{f}(y^n)=f(y^n)
	}}\frac{1}{K}P_{X^n|Y^n}({\cal T}_{\bX|\bU \bY}(u_{f(y^n)}^n,y^n))\nonumber \\
	& \geq(n+1)^{-|\cU ||\cX ||\cY |}2^{-nD(P_{\bX|\bU \bY}||P_{X|Y}|P_{\bU \bY})},
\end{align}
where $K = \sum_{i=1}^{2^{nR_2}} |f^{-1}(i)|$ is the number of $y^n\in{\cal T}_{\bY}$ satisfying
$\tilde{f}(y^n)=f(y^n)$. Also, since the decoding method searches the unique sequence $\hat{x}^n$ among $\cT_{\tX|\tU}(u^n_{f(y^n)})$ such that the bin index coincides with $x^n$, $P_{\textnormal{c},3}$ can be evaluated as
\begin{align}
	P_{\textnormal{c},3} & =1-\sum_{(x^n,y^n)}\Pr\{(X^n,Y^n)=(x^n,y^n)|E_{1}\land E_{2}\}\nonumber \\
	& \hphantom{-----.-}\cdot\Pr\{\exists\hat{x}^n\neq x^n ~~\mathrm{s.t.}~~ F(\hat{x}^n)=F(x^n),\hat{x}^n\in\cT _{\bX|\bU }(u_{f(y^n)}^n)\}\nonumber \\
	& \geq1-\sum_{(x^n,y^n)}\Pr\{(X^n,Y^n)=(x^n,y^n)|E_{1}\land E_{2}\}\cdot\sum_{\substack{
			\hat{x}^n\in\cT _{\bX|\bU }(u_{f(y^n)}^n)\\
			\hat{x}^n\neq x^n}}
	\Pr\{F(\hat{x}^n)=F(x^n)\}\nonumber \\
	& \geq1-\sum_{(x^n,y^n)}\Pr\{(X^n,Y^n)=(x^n,y^n)|E_{1}\land E_{2}\}\cdot\sum_{\substack{
			\hat{x}^n\in\cT _{\bX|\bU }(u_{f(y^n)}^n)\\
			\hat{x}^n\neq x^n}}
	2^{-nR_1}\nonumber \\
	& \geq1-2^{-n(R_{1}-H(\bX|\bU))}\nonumber \\
	& \geq1-2^{-n\epsilon},
\end{align}
where we use \cite[Lemma 2.5]{CK11}, and a condition $R_{1}-\epsilon\geq H(\bX|\bU )$
in the last inequality. Therefore, $P_{\textnormal{c}}$ can be evaluated
as
\begin{align}
	P_{\textnormal{c}} &\geq\prod_{i=1}^{3}P_{\textnormal{c},i}\nonumber \\
	& \geq(n+1)^{-|\cY |}2^{-n(D(P_{\bY}||P_{Y})+|I(\bU \land\bY)-R_{2}+\delta_{n}|^{+})}\nonumber \\
	& \quad\cdot(n+1)^{-|\cU ||\cX ||\cY |}2^{-nD(P_{\bX|\bU \bY}||P_{X|Y}|P_{\bU \bY})}(1-2^{-n\epsilon})\nonumber \\
	& =(n+1)^{-(|\cU ||\cX |+1)|\cY |}(1-2^{-n\epsilon})\nonumber \\
	& \quad\cdot2^{-n(D(P_{\bU \bX\bY}||P_{\bU |\bY}P_{XY})+|I(\bU \land\bY)-R_{2}+\delta_{n}|^{+})},\label{eq:d_r1}
\end{align}
where we used the following equality
\begin{align*}
	&D(P_{\bU\bX\bY}||P_{\bU|\bY}P_{XY}) = D(P_{\bX|\bU\bY}||P_{X|Y}|P_{\bU\bY})+D(P_{\bY}||P_Y),
\end{align*}
which is derived by the definition of the conditional KL-divergence.

Thus, by rearranging (\ref{eq:d_r1}), we have
\begin{align}
	\frac{1}{n}\log\frac{1}{P_{\textnormal{c}}}\leq D(P_{\bU \bX\bY}||P_{\bU |\bY}P_{XY})\!+\!|I(\bU \land\bY)-R_{2}+\delta_{n}|^{+} \nonumber \\
	+(|\cU ||\cX |+1)|\cY |\frac{\log(n+1)}{n}\!+\!\frac{1}{n}\log\frac{1}{1-2^{-n\epsilon}}.
\end{align}

Consequently, by taking the limit of $n$ and by the continuity of
the information quantities, we obtain
\begin{align}
G(R_{1},R_{2}|P_{XY}) \leq F(R_{1}-2\epsilon,R_{2}|P_{XY}).
\end{align}
Finally, by the continuity of $F$ with respect to $R_{1}$, we obtain
the desired result.

\subsection{Converse Part}
In this section, we prove the converse part by using the change of measure argument.
Define $\mathcal{C}=\{(x^n,y^n):\psi(\varphi_{1}(x^n),\varphi_{2}(y^n))=x^n\}$ as the set of sequences without error, and define the changed probability measure
\begin{align}
	P_{\tX ^n\tY ^n}(x^n,y^n) & =\Pr\{(X^n,Y^n)\!=\!(x^n,y^n):(X^n,Y^n)\!\in\!\mathcal{C}\}\nonumber \\
	& =\frac{P_{X^nY^n}(x^n,y^n)}{P_{X^nY^n}(\mathcal{C})}\1[(x^n,y^n)\in\mathcal{C}],
\end{align}
where $\1$ is the indicator function. Then, we can verify
\begin{equation}
    D(P_{\tX ^n\tY ^n}||P_{X^nY^n})=\log\frac{1}{P_{\textnormal{c}}}.\label{eq:c_basic}
\end{equation}

First, define random variables $\tMo=\varphi_{1}(\tX ^n)$ and $\tMt=\varphi_{2}(\tY ^n)$.
Since $nR_{1}\geq\log|\mathcal{C}|\geq H(\tMo)$, we can write a chain of inequalities
\begin{align}
	nR_{1} & \geq H(\tMo)\nonumber \\
	& \geq H(\tMo|\tMt)\nonumber \\
	& =H(\tMo|\tMt)-H(\tMo|\tMt,\tX ^n)\nonumber \\
	& =I(\tMo\land\tX ^n|\tMt)\nonumber \\
	& =H(\tX ^n|\tMt)-H(\tX ^n|\tMo,\tMt)\nonumber \\
	& =H(\tX ^n|\tMt). \label{eq:R1_1}
\end{align}
Here, note that $H(\tMo|\tMt,\tX ^n)=0$ since $\tMo$ is generated from $\tX ^n$; and $H(\tX ^n|\tMo,\tMt)=0$ since $\psi$ observes $\tMo$ and $\tMt$, then outputs the error-free estimation of $\tX ^n$.
Also, we consider $H(\tX ^n|\tMt)$.
Denoting $X_{j}^{-}=(X_{1},\dots,X_{j-1})$, we obtain
\begin{align}
	H(\tX ^n|\tMt) & =\sum_{j=1}^nH(\tX _{j}|\tMt,\tX _{j}^{-})\nonumber \\
	& \geq\sum_{j=1}^nH(\tX _{j}|\tMt,\tX _{j}^{-},\tY _{j}^{-})\nonumber \\
	& =nH(\tX _{J}|\tU _{J},J),\label{R1_2}
\end{align}
where $J\sim\mathtt{unif}(\{1,\dots,n\})$ is the time-sharing random variable and is assumed to be independent of all the other random variables involved.
Also, we put $\tU _{J}=(\tMt,\tX _{J}^{-},\tY _{J}^{-})$.

Next, since $I(\tMt\land\tX ^n|\tY ^n)=0$, we have
\begin{align}
	D(P_{\tX ^n\tY ^n}||P_{X^nY^n})
	& =D(P_{\tX ^n\tY ^n}||P_{X^nY^n})+I(\tMt\land\tX ^n|\tY ^n)\nonumber \\
	& =D(P_{\tX ^n\tY ^n}||P_{X^nY^n})+H(\tX ^n|\tY ^n)-H(\tX ^n|\tY ^n,\tMt).
\end{align}
By applying \cite[Proposition 1]{TW20}, we obtain
\begin{align}
	D(P_{\tX ^n\tY ^n}||P_{X^nY^n})+H(\tX ^n|\tY ^n) \geq n(D(P_{\tX _{J}\tY _{J}}||P_{XY})+H(\tX _{J}|\tY _{J})).
\end{align}
Also,
\begin{align}
	H(\tX ^n|\tY ^n,\tMt) & =\sum_{j=1}^nH(\tX _{j}|\tY ^n,\tMt,\tX _{j}^{-})\nonumber \\
	& \leq\sum_{j=1}^nH(\tX _{j}|\tY _{j},\tMt,\tX _{j}^{-},\tY _{j}^{-})\nonumber \\
	& =nH(\tX _{J}|\tY _{J},\tU _{J},J).
\end{align}
Therefore, we obtain
\begin{align}
	D(P_{\tX ^n\tY ^n}||P_{X^nY^n}) \geq n(D(P_{\tX _{J}\tY _{J}}||P_{XY})+I(\tU _{J},J\land\tX _{J}|\tY _{J})).\label{eq:d_1}
\end{align}

Moreover, since $nR_{2}\geq H(\tMt)$, we have
\begin{align}
	nR_{2} & \geq H(\tMt)\nonumber \\
	& =H(\tMt)-H(\tMt|\tX ^n,\tY ^n)\nonumber \\
	& =I(\tMt\land\tX ^n,\tY ^n)\nonumber \\
	& =H(\tX ^n,\tY ^n)-H(\tX ^n,\tY ^n|\tMt).
\end{align}
Then, we can consider the following inequality
\begin{align}
	D(P_{\tX ^n\tY ^n}||P_{X^nY^n})
	& \geq D(P_{\tX ^n\tY ^n}||P_{X^nY^n})+H(\tX^n,\tY^n)-H(\tX^n,\tY^n|\tMt)-nR_{2}.
\end{align}
A very similar argument in \cite[Proposition 1]{TW20} leads to
\begin{align}
	&D(P_{\tX^n\tY^n}||P_{X^nY^n})+H(\tX ^n,\tY ^n) =n(D(P_{\tX _{J}\tY _{J}}||P_{XY})+H(\tX _{J},\tY _{J})).
\end{align}
Also,
\begin{align}
	H(\tX ^n,\tY ^n|\tMt) & =\sum_{j=1}^nH(\tX _{j},\tY _{j}|\tMt,\tX _{j}^{-},\tY _{j}^{-})\nonumber \\
	& =nH(\tX _{J},\tY _{J}|\tU _{J},J).
\end{align}
Therefore, we obtain
\begin{align}
	\frac{1}{n}D(P_{\tX ^n\tY ^n}||P_{X^nY^n})
	&\geq D(P_{\tX _{J}\tY _{J}}||P_{XY})+H(\tX _{J},\tY _{J})-H(\tX _{J},\tY _{J}|\tU _{J},J)-R_{2}\nonumber \\
	&=D(P_{\tX _{J}\tY _{J}}||P_{XY})+I(\tU _{J},J\land\tX _{J},\tY _{J})-R_{2}\nonumber \\
	&=D(P_{\tX _{J}\tY _{J}}||P_{XY})+I(\tU _{J},J\land\tX _{J}|\tY _{J})+I(\tU _{J},J\land\tY _{J})-R_{2}. \label{eq:c_r2}
\end{align}

Combining (\ref{eq:c_basic}), (\ref{eq:d_1}) and (\ref{eq:c_r2}), we have
\begin{align}
	\frac{1}{n}\log\frac{1}{P_{\textnormal{c}}}\geq D(P_{\tX _{J}\tY _{J}}||P_{XY})+I(\tU _{J},J\land\tX _{J}|\tY _{J})+\left|I(\tU _{J},J\land\tY _{J})-R_{2}\right|^{+}.
\end{align}
Noting (\ref{eq:R1_1}) and (\ref{R1_2}), we have a condition that $R_{1}\geq H(\tX_{J} |\tU _{J},J)$. As a consequence, we obtain
\begin{align}
	\min_{\psi,\varphi_{1},\varphi_{2}}\frac{1}{n}\log\frac{1}{P_{\textnormal{c}}}
	\geq\min_{P_{\tU \tX \tY }}\left\{D(P_{\tX \tY }||P_{XY})+I(\tU \land\tX |\tY )+\left|I(\tU \land\tY )-R_{2}\right|^{+}\right\},
\end{align}
where the minimization of $P_{\tU \tX \tY }$ is taken under the condition $R_{1}\geq H(\tX |\tU )$.
Note that we can see
\begin{align}
	D(P_{\tX \tY }||P_{XY})+I(\tU \land\tX |\tY ) &=D(P_{\tU \tX \tY }||P_{\tU |\tY }P_{XY}).
    \label{eq:divandI}
\end{align}

Consequently, we have
\begin{align}
G(R_{1},R_{2}|P_{XY}) \geq F(R_{1},R_{2}|P_{XY}),
\end{align}
except the cardinality bound, which will be discussed in the next section.

\subsection{Cardinality Bound}

By the support lemma \cite[Lemma 15.4]{CK11}, we can restrict the
cardinality of $U$ to $|\cU |\leq|\cX ||\cY |+2$ as follows.
We set the following functions on ${\cal P}(\cX\times\cY)$:
\begin{align*}
	g_{1}(P_{\tX\tY}) & =H(\tX|\tY),\\
	g_{2}(P_{\tX\tY}) & =H(\tX),\\
	g_{3}(P_{\tX\tY}) & =H(\tY).
\end{align*}
Then, observe that
\begin{align*}
	P_{\tX \tY }(x,y) & =\sum_{u}P_{\tU }(u)P_{\tX\tY|\tU}(x,y|u),\\
	H(\tX |\tU ,\tY ) & =\sum_{u}P_{\tU }(u)g_{1}(P_{\tX\tY|\tU}(\cdot,\cdot|u)),\\
	H(\tX |\tU ) & =\sum_{u}P_{\tU }(u)g_{2}(P_{\tX\tY|\tU}(\cdot,\cdot|u)),\\
	H(\tY |\tU ) & =\sum_{u}P_{\tU }(u)g_{3}(P_{\tX\tY|\tU}(\cdot,\cdot|u)).
\end{align*}
Note that $I(\tU \land\tX |\tY )=H(\tX |\tY )-H(\tX |\tU ,\tY) $ and $I(\tU \land\tY )=H(\tY )-H(\tY |\tU )$.
Therefore, by the support lemma, it suffices to take $|\cU |\leq|\cX ||\cY |-1+3=|\cX ||\cY |+2$.

\section{Application to Privacy Amplification for Bounded Storage Eavesdropper} \label{sec:pa}
Now we consider the privacy amplification that is a technique to distill a secret key from a random variable by a (possibly random) function so that the distilled key and eavesdropper's random variable are statistically independent \cite{BBR88}, \cite{BBM96}.

It is well known that there is duality between source coding problems and random number generation (e.g. \cite{Han02}).
\cite{WO12} introduced a privacy amplification problem for bounded storage eavesdropper.
In this situation, an eavesdropper can access a data that is correlated with a secret key and stored in a bounded sized storage.
This problem is interesting as a dual problem for the WAK source coding problem and in order to analyze its security, we need an analysis that is equivalent to the strong converse for the WAK network.
Also, they analyzed the security using the results regarding the strong converse for the WAK network by \cite{AGK76}.
\cite{SO19} analyzed the security using the analysis method of the exponential strong converse for the WAK network by \cite{Oo19}.
In this section, we provide our derivation of the tight exponent for the WAK network for a security analysis of this problem.

\subsection{Privacy Amplification for Bounded Storage Eavesdropper}
\begin{figure}[h]
	\centering
	\includegraphics[scale=0.9]{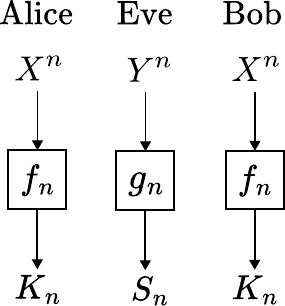}
	\caption{Privacy amplification for bounded storage eavesdropper\protect\label{fig:PA}}
\end{figure}

In this section, we introduce the privacy amplification for bounded storage eavesdropper illustrated in Fig. \ref{fig:PA}.
Let $(X^n,Y^n)$ be an i.i.d. correlated sources on $\cX^n \times \cY^n$ and the alphabets $\cX$ and $\cY$ are finite sets.
Suppose that Alice (a sender) and Bob (a receiver) have a random variable $X^n$ on a set $\cX^n$ and Eve (an eavesdropper) has a random variable $Y^n$ on $\cY^n$ that is correlated with $X^n$.
Note that $Y^n$ is encoded and stored in a storage, since the storage does not have enough space to store $Y^n$ directly.
In the privacy amplification, a secret key $K_n$ is distilled by a function
\begin{align*}
	f_n:\cX^n \to \cK_n = \{0,1\}^{nR_1}
\end{align*}
so that they and Eve's information $S_n$ are statically independent and the key $K_n$ is uniformly distributed on the key alphabet $\cK_n$.
\begin{comment}
The joint probability distribution of the key and Eve's information is given by
\begin{align*}
	P_{K_nS_n}(k,s) = \sum_{x^n \in f_n^{-1}(k)} P_{X^nS_n}(x^n,s),
\end{align*}
for $(k,s) \in \cK_n \times \cS_n$, where $f_n^{-1}(k) = \{x^n:f_n(x^n) = k\}$.
\end{comment}
Eve's random variable $S_n$ is stored in a storage, which is obtained as a function value of i.i.d. source $Y^n$ through a function, and the rate of the storage size is bounded.
Eve has a function value $E_n$ of some function
\begin{align*}
	g_n: \cY^n \to \cS_n = \{0,1\}^{nR_2}.
\end{align*}
For privacy amplification, Alice and Bob try to distill a secret key as long as possible.
In this paper, since we use the total variational distance as a security criterion, we require the quantity
\begin{align*}
	\Delta = d_{\textnormal{var}}(P_{K_nS_nF},P_{\textnormal{unif}}\times P_{S_n}\times P_{F})
\end{align*}
to be small, where $P_{\textnormal{unif}}$ is the uniform distribution on $\cK_n$ and $F$ is a random mapping $F:\cX^n \to \cK_n$.
Usually, we choose $F$ from the universal hash family (e.g. see \cite[Chapter 7]{TW23}).

\subsection{Privacy Amplification and WAK Network}
In this section, we evaluate the security criterion for the privacy amplification for bounded storage eavesdropper by using our derivation of the tight exponent for the WAK network.
Before we state our argument, we give a well-known lemma regarding the security criterion \cite{Ren08} (see also \cite[Chapter 7]{TW23}).
\begin{lem} \label{lem:pa}
	For a given distribution $P_{XS_n}$ and a number $\tau \in \mathbb{R}$, let
	\begin{align*}
		\tilde{\cT} = \left\{(x^n,s) : \frac{1}{P_{X^n|S_n}(x^n|s_n)} \geq \tau\right\}.
	\end{align*}
	Then, we have
	\begin{align}
		\Delta &\leq \Pr\left(	\log \frac{1}{P_{X^nS_n}(X^n|S_n)} < \tau\right)+\frac{1}{2} \sqrt{2^{-\tau+nR_1}}.
		\label{eq:PA_lemma}
	\end{align}
\end{lem}

One of the goals of privacy amplification is to evaluate a security criterion.
We are interested in the trade-off between the key rate $R_1$ and Eve's storage size $R_2$ under the condition that the key $K_n$ and Eve's information $S_n$ are statistically independent.
Now, by using our derivation of the strong converse exponent, we evaluate the first term on the right-hand side of \eqref{eq:PA_lemma} for the privacy amplification illustrated in Fig. \ref{fig:PA}.
In the following theorem, we show the evaluation for the security criterion with our exponent.

\begin{thm} \label{thm:pa}
	Setting $\tau = nR_1'$ where $R_1' = R_1+\delta$ for a positive number $\sigma$, the security criterion $\Delta$ can be evaluated as
	\begin{align}
		\Delta \leq 2^{-nF(R_1',R_2|P_{XY})}+\frac{1}{2}\sqrt{2^{-n\delta}},
	\end{align}
    where $F(R_1',R_2|P_{XY})$ is defined by \eqref{eq:formula_R}.
    Recall that the function $F(R_1',R_2|P_{XY})$ is positive if and only if $(R_1',R_2) \notin \mathcal{R}_{\textnormal{WAK}}$.
\end{thm}

\begin{proof}
In this proof, we evaluate $\Pr(\log (1/P_{X^n|S_n}(X^n|S_n)) < nR_1')$.
First, let $P'_{\textnormal{c}}$ be a number such that
\begin{align}
	P'_{\textnormal{c}} =  \Pr\left(\frac{1}{n}\log \frac{1}{P_{X^n|S_n}(X^n|S_n)} < R_1'\right).
\end{align}
Also, we define
\begin{align}
	\cT ^c &= \left\{
	(x^n,y^n) : \frac{1}{n}\log \frac{1}{P_{X^n|S_n}(x^n|g_n(y^n))}< R_1'
	\right\}, \\
	P_{\tX^n \tY^n} &= \Pr((X^n,Y^n) = (x^n,y^n) | (X^n,Y^n) \in \cT^c) \nonumber\\
	& =\frac{P_{X^nY^n}(x^n,y^n)}{P_{X^nY^n}(\cT^c)}\1[(x^n,y^n)\in \cT^c].
\end{align}
Then, we can calculate
\begin{equation}
	D(P_{\tX ^n\tY ^n}||P_{X^nY^n})=\log\frac{1}{P'_{\textnormal{c}}},
\end{equation}
which corresponds to \eqref{eq:c_basic} of the converse part in Section \ref{sec:proof}.
This means that our converse proof method for the exponent can be also applied to the consideration of the privacy amplification for bounded storage eavesdropper.

Next, defining a random variable $\tilde{S}_n = g_n(Y^n)$, we evaluate $R_1'$ as
\begin{align}
nR_1' &\geq \mathbb{E}\left[\log \frac{1}{P_{X^n|S_n}(\tX^n|\tilde{S}_n)}\right] \nonumber\\
&= \mathbb{E}\left[\log \frac{1}{P_{\tX^n|\tilde{S}_n}(\tX^n|\tilde{S}_n)}\right]+\mathbb{E}\left[\log \frac{P_{\tX^n|\tilde{S}_n}(\tX^n|\tilde{S}_n)}{P_{X^n|S_n}(\tX^n|\tilde{S}_n)}\right] \nonumber\\
&= H(\tX^n|\tilde{S}_n)+D(P_{\tX^n|\tilde{S}_n}||P_{X^n|S_n}|P_{\tilde{S}_n}) \nonumber\\
&\geq H(\tX^n|\tilde{S}_n),
\end{align}
where $\mathbb{E}[\cdot]$ denotes the expectation of $P_{\tX^n|\tilde{S}_n}$.
Single-letterization in the similar way as in Section \ref{sec:proof} yields
\begin{align}
    R_1' \geq H(\tX|\tU),
\end{align}
where $\tU$ is a similar random variable to the one that appears in the process of single-letterization in the converse of Section \ref{sec:proof}.
Also, in the same manner, we can obtain
\begin{align}
    R_2 \geq I(\tY \land \tU).
\end{align}

According to the proof of the converse for our main theorem, we have
\begin{align}
	P'_{\textnormal{c}} \leq 2^{-nF(R_1',R_2|P_{XY})} .
\end{align}
Thus, for any $P_{XY}$, we obtain
\begin{align}
	\Pr\left(\frac{1}{n}\log \frac{1}{P_{X^n|S_n}(X^n|S_n)}<R_1' \right) \leq 2^{-nF(R_1',R_2|P_{XY})}.
\end{align}
As a conclusion, by Lemma \ref{lem:pa}, the security criterion for the privacy amplification for a bounded storage eavesdropper $\Delta$ can be bounded as
\begin{align}
	\Delta \leq 2^{-nF(R_1',R_2|P_{XY})}+\frac{1}{2}\sqrt{2^{-n\delta}}.
\end{align}
\end{proof}

When $(R_1,R_2)$ is not in the achievable rate region for the WAK network, we can take $\delta > 0$ so that  $(R_1',R_2)$ is not achievable as well.
Then, Theorem \ref{thm:pa} states that the security criterion $\Delta$ exponentially converges to $0$.

\begin{comment}
Since $nR_1 \geq \mathbb{E}\left[\log\frac{1}{P_{X^n|S_n}(\hat{X}^n|\hat{S}_n)}\right]$, we can write a chain of inequalities
\begin{align}
	nR_1 &\geq \mathbb{E}\left[\log\frac{1}{P_{X^n|S_n}(\hat{X}^n|\hat{S}_n)}\right] \nonumber\\
	&= \mathbb{E}\left[\log\frac{1}{P_{\hat{X}^n|\hat{S}_n}(\hat{X}^n|\hat{S}_n)}\right]+\mathbb{E}\left[\log\frac{P_{\hat{X}^n|\hat{S}_n}(\hat{X}^n|\hat{S}_n)}{P_{X^n|S_n}(\hat{X}^n|\hat{S}_n)}\right] \nonumber\\
	&= H(\hat{X}^n|\hat{S}_n)+D(P_{\hat{X}^n|\hat{S}_n}||P_{X^n|S_n}|P_{\hat{S}_n}) \nonumber\\
	&\geq H(\hat{X}^n|\hat{S}_n).
	\label{eq:paR1}
\end{align}
Also, by taking a similar approach to the strong converse for the WAK problem, we can get
\begin{align}
	nR_2 \geq I(\hat{S}_n \land \hat{Y}^n).
	\label{eq:paR2}
\end{align}
Then, we single-letterize \eqref{eq:paR1} and \eqref{eq:paR2} by using almost the same technique in Section \ref{sec:proof}, and we have
\begin{align}
	R_1 &\geq H(\hat{X}|\hat{S}), \\
	R_2 &\geq I(\hat{S} \land \hat{Y}),
\end{align}
where we denote $\hat{S} = \psi(\hat{Y})$.
As a consequence, we obtain
\begin{align}
	content...
\end{align}
\end{comment}

\section{Conclusion}
In this paper, we derived the tight strong converse exponent for the WAK network.
The rate region of the WAK network has two characteristics: an auxiliary random variable and the Markov chain.
Because of these characteristics, deriving the tight strong converse exponent of the WAK network has been regarded as a challenging problem.
Also, we obtained that the numerical experiment suggests that a better result can be obtained when the Markov constraint does not hold for the arguments of the optimization function $P_{\tU\tX\tY}$.
We believe that our proof methods can be also applied for other multi-terminal network problems.

\section*{Acknowledgments}
The authors would like to thank Prof. Yasutada Oohama and Prof. Jun Chen for a valuable discussion.
The authors also would like thank Prof. Tomohiko Uyematsu for suggesting to consider the special case in Section \ref{sec:comp}.

\printbibliography

\end{document}